\documentclass{endm}
\usepackage{endmmacro}
\usepackage{graphicx}
\usepackage{amsfonts}
\usepackage{amsmath}
\usepackage{amssymb}
\usepackage{complexity}
\usepackage[american]{babel}
\usepackage{subcaption}
\usepackage{color}
\usepackage{fancyhdr}

\begin{document}

% DO NOT REMOVE: Creates space for Elsevier logo, ScienceDirect logo
% and ENDM logo
\begin{verbatim}\end{verbatim}\vspace{2.5cm}

\begin{frontmatter}

\title{The tessellation problem of quantum walks$^\star$}
\thanks{Email: {\texttt {\normalshape
        \{santiago, lfignacio, celina, franklin,posner\}@cos.ufrj.br}},
        {\texttt {\normalshape\{tharsodf, portugal\}@lncc.br} },
        {\texttt {\normalshape
        luis@ic.uff.br}}. We acknowledge support from CAPES, CNPq and FAPERJ.}
%\thanks{Correponding author: {\normalshape {\tt celina@cos.ufrj.br}}. We acknowledge important discussions with D. Posner and support from CAPES, CNPq and FAPERJ.}

\author[pesc]{A. Abreu},
\author[pesc]{L. Cunha},
\author[lncc,ufes]{T. Fernandes},
\author[pesc]{C. de Figueiredo},
\author[uff]{L. Kowada},
\author[pesc]{F. Marquezino},
\author[pesc]{D. Posner},
\author[lncc]{R. Portugal}  %email portugal@lnccc.br

%\author[pesc]{A. Abreu},
%\author[pesc]{L. Ign\'acio},
%\author[lncc,ufes]{T. Fernandes},
%\author[pesc]{C. de Figueiredo},
%\author[uff]{L. Kowada},
%\author[pesc]{F. Marquezino},
%\author[lncc]{R. Portugal}  %email portugal@lnccc.br

\address[pesc]{Universidade Federal do Rio de Janeiro, Brazil}
\address[lncc]{Laborat\'orio Nacional de Computa\c{c}\~ao Cient\'{\i}fica, Brazil}
\address[uff]{Universidade Federal Fluminense, Brazil}
\address[ufes]{Universidade Federal do Esp\'{\i}rito Santo, Brazil}

%\author{T. D. Fernandes\thanksref{}\thanksref{}}  %email tharsofernandes@gmail.com
%\address{Universidade Federal do Esp\'{\i}rito Santo\\ Alegre, Brazil}

%\author{My Co-author\thanksref{coemail}}
%\address{My Co-author's Department\\ My Co-author's University\\
%   My Co-author's City, My Co-author's Country} \thanks[ALL]{Thanks
%   to everyone who should be thanked} \thanks[myemail]{Email:
%   \href{mailto:myuserid@mydept.myinst.myedu} {\texttt{\normalshape
%   myuserid@mydept.myinst.myedu}}} \thanks[coemail]{Email:
%   \href{mailto:couserid@codept.coinst.coedu} {\texttt{\normalshape
%   couserid@codept.coinst.coedu}}}

\begin{abstract}
Quantum walks have received a great deal of attention recently because they can be used to develop new quantum algorithms and to simulate interesting quantum systems. In this work, we focus on a model called staggered quantum walk, which employs advanced ideas of graph theory and has the advantage of including the most important instances of other discrete-time models. The evolution operator of the staggered model is obtained from a tessellation cover, which is defined in terms of a set of partitions of the graph into cliques. It is important to establish 
%what is the minimum number of tessellations 
the minimum number of tessellations required in a tessellation cover,
and what classes of graphs admit a small number of tessellations. We describe two main results: (1)~infinite classes of graphs where we relate the chromatic number of the clique graph to the minimum number of tessellations required in a tessellation cover, and (2)~the problem of deciding whether a graph is $k$-tessellable for $k\ge 3$ is NP-complete.
\end{abstract}

% Quantum walks have received a great deal of attention recently because they can be used to develop new quantum algorithms and to simulate interesting quantum systems. In this work, we focus on a model called staggered quantum walk, which employs advanced ideas of graph theory and has the advantage of including the most important instances of other discrete-time models. The evolution operator of the staggered model is obtained from a tessellation cover, which is defined in terms of a set of partitions of the graph into cliques. It is important to establish 
%what is the minimum number of tessellations 
% the minimum number of tessellations required in a tessellation cover,
% and what classes of graphs admit a small number of tessellations. We describe two main results: (1)~The chromatic number of the clique graph plays a central role in determining the minimal number of tessellations, and (2)~the problem of deciding whether a graph is $k$-tessellable for $k\ge 3$ is NP-complete.


\end{frontmatter}

\newpage

\section{Introduction}
\vspace{-0.3cm}
Random walks play an important role in computer science mainly in the area of algorithms and it is expected that quantum walks, which is the quantum counterpart of random walks, will play at least a similar role in quantum computation. In fact, in the last decades there has been much interest in the area of quantum walks and they are considered one of the main techniques to build algorithms for quantum computers~\cite{Ven12}. 

Quantum walks on graphs come in two flavors: continuous- and discrete-time. There are more than one model for discrete-time quantum walks. The most known ones are the coined~\cite{Aharonov:2000} and Szegedys's model~\cite{Szegedy:2004}. Recently, a new model called staggered quantum walk~\cite{Por16b,PSFG16} was proposed, which in some sense is more general than the previous ones because the staggered model includes the entire Szegedy's and the most interesting instances of the coined model.

The staggered model on graphs is defined by an evolution operator that is a product of unitary matrices obtained from the graph by a tessellation process. A tessellation is a partition of the graph into cliques so that the union of the cliques covers the vertex set but not necessarily the edge set. There is a recipe to build a unitary and Hermitian matrix based on a chosen tessellation~\cite{Por16b,PSFG16}. To define the evolution operator of the quantum walk, one has to choose extra tessellations until the tessellation union covers the edge set.

The simplest evolution operators are the product of few unitary matrices, and at least two matrices  corresponding to 2-tessellable graphs are required. The class of 2-tessellable graphs was exhaustively studied in Ref.~\cite{Por16b}, which showed that a graph is 2-tessellable if and only if its clique graph is 2-colorable. Clique graphs~\cite{Szw03} play a central role in the tessellation problem. Graphs whose clique graphs require $k>2$ colors to be vertex-colored cannot be covered by just two tessellatons, for example, the stars $S_k$.

%require $k>2$ colors to be vertex-colored with just two colors cannot be covered by just two tessellations, for example, the stars $S_k$.

An important concept in this context is the minimum number of tessellations required to cover the graph, called the tessellation number. An upper bound is easily established by using the clique graph operator. In fact, we show that the chromatic number of the clique graph is a tight upper bound for the tessellation number. On the other hand, we obtain a class of 3-tessellable graphs whose clique graphs have arbitrarily large chromatic number, establishing the tightness of the lower bound.

Another important concept is the NP-completeness character of the tessellation problem. We show that to determine whether a graph is $k$-tessellable for $k>2$ is NP-complete by reduction from the edge coloring problem of triangle free graphs with degree at most three~\cite{Kor97}. The tessellation problem is also directly related to the set cover problem~\cite{GareyJohnson}.

%%% To find examples of $k$-tessellable graphs whose clique graphs have the largest chromatic number, we performed computer experiments employing greedy algorithms for the set cover problem. 

% \vspace{-1.6cm}

%\vspace{-0.7cm}
\subsubsection*{Preliminaries}

A \textbf{clique} is a subset of vertices of a graph such that its induced subgraph is complete. A clique of size $d$ is called a $d$-\textbf{clique}. In a \textbf{partition} of the graph into cliques, each element of the partition is a clique and two elements of the partition cannot have a vertex in common. 

\begin{definition}
A \textbf{tessellation} ${\mathcal{T}}$ is a partition of the graph into cliques, where each clique is called a \textbf{polygon} (or a \textbf{cell}), such that the union of the polygons covers the vertex set. An edge \textbf{belongs} to the tessellation if and only if both endpoints of the edge belong to the same polygon. The set of edges belonging to ${\mathcal{T}}$ is denoted by ${\mathcal{E}}({\mathcal{T}})$.
\end{definition}

\begin{definition}
A \textbf{tessellation cover} of size $k$ of a graph $\Gamma$ is a set of $k$~tessellations ${\mathcal{T}}_1,...,{\mathcal{T}}_k$ whose union $\cup_{i=1}^k\,{\mathcal{E}}({\mathcal{T}}_i)$ is the edge set of $\Gamma$. The \textbf{tessellation number} $T(\Gamma)$ is the cardinality of a smallest tessellation cover of~$\Gamma$. A graph~$\Gamma$ is called \textbf{$t$-tessellable} for an integer $t$ when $T(\Gamma)\le t$. The $t$-\textbf{tessellation} problem asks given a graph $\Gamma$ whether $\Gamma$ is $t$-tessellable. %$T(\Gamma)\leq t$.
\end{definition}

As an example, we illustrate the above definitions using the \textbf{wheel graph}, which is the graph $W_{n}$ for $n>2$ with vertex set $\{0,1,2,...,$ $n\}$ and edge set $\{\{0,n\},\{1,n\},...,$ $\{n-1,n\},$ $\{0,1\},$ $\{1,2\},... ,\{n-2,n-1\},\{n-1,0\}\}$. Notice that $T(W_{6})=3$ because tessellations ${\mathcal{T}}_0=\{\{0\},\{3\},\{1,2\},\{4,5,6\}\}$, ${\mathcal{T}}_1=\{\{1\},\{4\},\{0,5\},\{2,3,6\}\}$, ${\mathcal{T}}_2=\{\{2\},\{5\},\{3,4\},\{0,1,6\}\}$ form a minimum tessellation cover. In fact, each of ${\mathcal{T}}_0, {\mathcal{T}}_1$, and ${\mathcal{T}}_2$ is a partition into cliques which covers all vertices. The set ${\mathcal{E}}({\mathcal{T}}_0)\cup{\mathcal{E}}({\mathcal{T}}_1)\cup{\mathcal{E}}({\mathcal{T}}_2)$ is equal to the edge set of $W_6$, and it is not possible to cover the edge set with less than three tessellations because a tessellation of $W_6$ can cover at most two edges incident to vertex $6$.

\begin{definition}
A \textbf{coloring} (resp. an \textbf{edge-coloring}) of a graph is a labeling of vertices (edges) with colors such that no two adjacent vertices (incident edges) have the same color. A $k$-\textbf{colorable} ($k$\textbf{-edge-colorable}) graph is the one whose vertices (edges) can be colored with at most $k$ colors so that no two adjacent vertices (edges) share the same color. The \textbf{chromatic number} $\chi(\Gamma)$ (\textbf{chromatic index} $\chi'(\Gamma)$) of a graph $\Gamma$ is the smallest number of colors needed to color the vertices (edges) of $\Gamma$.
\end{definition}

% \vspace{-0.8cm}
 \vspace{-0.3cm}
\section{Results}

\vspace{-0.4cm}
\subsubsection*{How far from $\chi(K(\Gamma))$ is $T(\Gamma)$?} 
In this subsection, we relate the tessellation number to some non-trivial classes of graphs and we prove the following proposition.

%and try to formulate a conjecture on the lower bound. An upper bound on the tessellation number of a graph is obtained by a coloring of its clique graph $K(\Gamma)$. 

\begin{proposition} \label{prop:upper}
Let $\Gamma$ be a graph whose clique graph is not 2-colorable. Then, $3\le T(\Gamma) \le \chi\big(K(\Gamma)\big).$
\end{proposition}
\vspace{-0.4cm}\begin{proof} 
The lower bound $3\le T(\Gamma)$ is a direct consequence of the fact that a graph is 2-tessellable if and only if its clique graph is 2-colorable, proved in Ref.~\cite{Por16b}. Now we give a proof for the upper bound $T(\Gamma) \le \chi\big(K(\Gamma)\big)$. We define a family of $\chi(K(\Gamma))$ tessellations whose union covers the edges of graph $\Gamma$ as follows. Consider an optimal coloring of $K(\Gamma)$, and let $S_g$ be the set of maximal cliques corresponding to the vertices of $K(\Gamma)$ colored by color $g$. Any pair of such maximal cliques must be disjoint, so we can define a tessellation ${\mathcal{T}}_g$ whose polygons are the cliques of $S_g$ together with missing vertices of $\Gamma$. Since every edge of $\Gamma$ belongs to at least one maximal clique, the union of the $\chi(K(\Gamma))$ defined tessellations covers the edges of graph~$\Gamma$, as required in order to establish the upper bound.
\end{proof}

Now we show that both bounds of Proposition~\ref{prop:upper} are tight. Next Proposition shows that the upper bound is tight.

% Tight class wrt the upper bound
\begin{proposition}\label{prop:windmill}
 Let $\Gamma$ be the \textbf{windmill} graph with $\ell$ maximal cliques $C_1, C_2,\\ \ldots, C_{\ell}$ such that the intersection of $C_1, C_2, \ldots, C_{\ell}$ is precisely one vertex $u$. Then, $T(\Gamma) = \chi(\mathcal{K}(\Gamma)) = \ell$.
\end{proposition}
\vspace{-0.4cm}\begin{proof}
The clique graph of $\Gamma$ is the complete graph with $\ell$ vertices. Hence, by Proposition~\ref{prop:upper}, $\ell$ is an upper bound for $T(\Gamma)$. On the other hand, it is not possible to cover the edge set with less than $\ell$ tesselations, since each tessellation cannot cover the edges of more than one maximal clique $C_i$.
%, since a tessellation of $\Gamma$ can cover at most $|C_i|-1$ edges incident to vertex $u$.
\end{proof}

The tightness of the lower bound is revealed by the following class of graphs.

\begin{definition}
The \textbf{$(3,n)$-extended wheel graph} $E_{3,n}$ for $n\ge 2$ is defined by adding to the wheel graph $W_{3n}$ the following edges: $\{3i,3j\},\{3i+1,3j+1\}$ and $\{3i+2,3j+2\}$, for $0\le i<j< n$.
\end{definition}

% because it induces a triangle of the spanning wheel graph and in $E_{3,n}$ the set of vertices $\{i, 0 \leq i < 3n \}$ contains no $3$-clique
\begin{proposition}\label{lemma_a}
 The maximal cliques of $E_{3,n}$ are $3$-cliques or $(n+1)$-cliques. The number of maximal cliques is $3n+3$. The maximal cliques are the $3$-cliques of the spanning wheel $W_{3n}$, plus three new $(n+1)$-cliques. All maximal cliques share the vertex with label $3n$.
\end{proposition}
\begin{proof} 
Let $i$ be an index in the range $0\le i< 3n$ and the arithmetic with this index be performed modulo $3n$. Then, the set $\{i,i+1,3n\}$ is a maximal $3$-clique because it induces a triangle of the spanning wheel graph and in $E_{3,n}$ the set of vertices $\{i, 0 \leq i < 3n \}$ contains no $3$-clique. Now consider the three sets of vertices $\{0,3,6,...,3n-3,3n\}$, $\{1,4,7,...,3n-2,3n\}$, and  $\{2,5,8,...,3n-1,3n\}$, each of them with cardinality $(n+1)$. We claim that each one is a maximal $(n+1)$-clique. Consider the set  $\{0,3,6,...,3n-3,3n\}$ (analogous for the other ones). All vertices in this set are adjacent because the edges are either $\{3i,3j\}$ for some $0\le i,j<n$ or $\{3i,3n\}$ for some $0\le i<n$. In the first case, these edges were added to $W_{3n}$ to define $E_{3,n}$, and in the second case the edges belong to the spanning wheel graph. If a new vertex is added, the new vertex must have the form $3i+1$ or $3i+2$ for some $0\le i<n$ and it will not be adjacent to all vertices of set  $\{0,3,6,...,3n-3,3n\}$. Then, this set is a maximal clique. There are three such $(n+1)$-maximal cliques and there are no other maximal cliques. Then, the total number of maximal cliques is $3n+3$ and all of them share the vertex with label $3n$.%. It is straightforward to verify that there are no more maximal cliques. Then, the total number of maximal cliques is $3n+3$ and all of them share the vertex with label $3n$.
\end{proof}

%\begin{lemma}\label{lemma1}
%$\chi\big(K(E_{3,n})\big)=3 n + 3$ for $n\ge 2$.
%\end{lemma}
%\begin{proof} 

It follows from Proposition~\ref{lemma_a} that $K(E_{3,n})$ is the complete graph with $3n+3$ vertices, and hence Proposition~\ref{prop:upper} establishes the upper bound $T(E_{3,n})\le 3n+3$. We prove next that the actual value of $T(E_{3,n})$ is much smaller.

\begin{proposition}\label{lemma2}
$T(E_{3,n})=3$ for $n\ge 2$.
\end{proposition}
\begin{proof} Let us show that $E_{3,n}$ is $3$-tessellable by describing explicitly three tessellations that cover the edges of $E_{3,n}$. The tessellations are the following ones: 
\begin{itemize}
\item[] ${\mathcal{T}}_0=\{\{0,3,6,...,3n-3,3n\},\{3i+1, 3i+2\}\textrm{ for }\,0\le i\le n-1\}$, 
\item[] ${\mathcal{T}}_1=\{\{1,4,7,...,3n-2,3n\},\{3i+2, 3i+3\}\textrm{ for }\,0\le i\le n-1\}$, 
\item[] ${\mathcal{T}}_2=\{\{2,5,8,...,3n-1,3n\},\{3i+3, 3i+4\}\textrm{ for }\,0\le i\le n-1\}$,
\end{itemize}
where the arithmetic with index $i$ is performed modulo $3n$. Let us show that ${\mathcal{T}}_0$ is a well defined tessellation (analogous for the other ones) by checking each item of the following list: (1) Each polygon in ${\mathcal{T}}_0$ must be a clique, (2) the polygons in ${\mathcal{T}}_0$ must be pairwise disjoint, and (3) the union of the polygons in ${\mathcal{T}}_0$ must be the vertex set. By using the proof of Proposition~\ref{lemma_a} and the fact that $\{3i+1, 3i+2\}$ is an edge of the spanning wheel, we check item~(1). Using that set $\{0,3,6,...,3n-3,3n\}$ is comprised of vertices that are multiple of 3 while no vertex in sets $\{3i+1, 3i+2\}$ for $0\le i\le n-1$ is multiple of 3, we check item~(2). The union of the sets in ${\mathcal{T}}_0$ is the vertex set, and we check item~(3). Since no edge belongs to more than one tessellation and each tessellation covers $n\,(n+3)/2$ edges, the union ${\mathcal{E}}({\mathcal{T}}_0)\cup{\mathcal{E}}({\mathcal{T}}_1)\cup{\mathcal{E}}({\mathcal{T}}_2)$ covers $3n(n+3)/2$ edges, which is the number of edges of $E_{3,n}$. It is not possible to cover the edges of $E_{3,n}$ with less than three tessellations because the chromatic number of the clique graph of $E_{3,n}$ is larger than 2.  Then, $T\big(E_{3,n}\big)=3$ for $n\ge 2$.
\end{proof}

%It is straightforward to verify that

\vspace{-0.5cm}
\subsubsection*{$3$-tessellation is \NP-complete}

Deciding whether a graph is $k$-tessellable for $k\ge 3$ is NP-complete. To prove this statement, we use the class of triangle-free graphs with maximum vertex degree 3 because the 3-edge-coloring problem in this class is NP-complete~\cite{Kor97}.

\begin{theorem}
Deciding whether a graph is 3-tessellable is \NP-complete.
\end{theorem}
\vspace{-0.4cm}\begin{proof}
In a triangle-free graph $\Gamma$, a $3$-tessellation corresponds to 3 matchings of $\Gamma$ covering its edge set, and so define a 3-edge-coloring of $\Gamma$.
\end{proof}

%\section{Conclusions}
\vspace{-0.4cm}
\subsubsection*{Final remarks}

We have shown that the lower bound of $T(\Gamma)$ does not depend on $\chi\big(K(\Gamma)\big)$, at least when $\chi\big(K(\Gamma)\big)$ is multiple of 3. Besides classes $E_{3,n}$ and the windmill graphs for which we were able to establish tessellation numbers satisfying $T(\Gamma) = 3$ and $T(\Gamma) = \chi(K(\Gamma))$, respectively, we were able to define additional infinite classes of graphs satisfying $\frac{T(\Gamma)}{\chi(K(\Gamma))} = \frac{1}{2}, \frac{1}{3}$, $\frac{1}{4}$, $T(\Gamma)=\sqrt{\chi(K(\Gamma))}$, and further ones obtained by extending $E_{3,n}$ into $E_{k,n}$ for $k\ge 4$.

For every graph $\Gamma$, is there a minimum tessellation cover such that every tessellation contains a polygon which is a maximal clique of $\Gamma$?

% \begin{small}
\vspace{-0.4cm}
\bibliographystyle{endm}   
%\bibliography{tessnumber} 

\begin{thebibliography}{1}
\expandafter\ifx\csname url\endcsname\relax
  \def\url#1{\texttt{#1}}\fi
\expandafter\ifx\csname urlprefix\endcsname\relax\def\urlprefix{URL }\fi
\newcommand{\enquote}[1]{``#1''}

\vspace{-0.2cm}
\bibitem{Aharonov:2000}
Aharonov, D., Ambainis, A., Kempe, J. and Vazirani, U., \emph{Quantum walks on
  graphs}, in: \emph{Proceedings of the Thirty-third Annual ACM Symposium on
  Theory of Computing}, STOC '01 (2001), pp. 50--59.

\vspace{-0.2cm}
\bibitem{GareyJohnson}
Garey, M., and Johnson, D., \enquote{{Computers and Intractability: A Guide
  to the Theory of NP-Completeness},} W. H. Freeman \& Co., New York, USA,
  1979.

\vspace{-0.2cm}
\bibitem{Kor97}
Koreas, D., \emph{{The NP-completeness of chromatic index in triangle free
  graphs with maximum vertex of degree 3}}, Applied Mathematics and Computation
  \textbf{83} (1997), pp.~13 -- 17.

\vspace{-0.2cm}
\bibitem{Por16b}
Portugal, R., \emph{Staggered quantum walks on graphs}, Physical Review A
  \textbf{93} (2016), p.~062335.

\vspace{-0.2cm}
\bibitem{PSFG16}
Portugal, R., Santos, R., Fernandes, T., and Gon{\c{c}}alves, D.,
  \emph{The staggered quantum walk model}, Quantum Information Processing
  \textbf{15} (2016), pp.~85--101.

\vspace{-0.2cm}
\bibitem{Szegedy:2004}
Szegedy, M., \emph{Quantum speed-up of {M}arkov chain based algorithms}, in:
  \emph{Proceedings of the 45th Symposium on Foundations of Computer Science},
  2004, pp. 32--41.

\vspace{-0.2cm}
\bibitem{Szw03}
Szwarcfiter, J., \emph{A survey on clique graphs}, in: Reed, B. and Sales, C., editors, \emph{Recent Advances in Algorithms and Combinatorics},
  Springer, New York, 2003 pp. 109--136.

\vspace{-0.2cm}
\bibitem{Ven12}
Venegas-Andraca, S., \emph{Quantum walks: a comprehensive review}, Quantum
  Information Processing \textbf{11} (2012), pp.~1015--1106.

\end{thebibliography}

% \end{small}

%%%\end{document}

\newpage

\section*{Appendix: Example of some class of graphs $\Gamma$ and their tessellation number $T(\Gamma)$}

We should mention that in order to obtain some of the infinite classes of graphs presented below, we have used the following computational approach. Let ${\mathcal{S}}$ be the set of all tessellations $\{{\mathcal{T}}_1,...,{\mathcal{T}}_r\}$ of a graph $\Gamma$. Define the set ${\mathcal{E}}({\mathcal{S}})=\{{\mathcal{E}}({\mathcal{T}}_1),...,{\mathcal{E}}({\mathcal{T}}_r)\}$. A minimum tessellation cover of $\Gamma$ corresponds to a minimum set cover~\cite{GareyJohnson} for ${\mathcal{E}}({\mathcal{S}})$. This approach for finding the minimum tessellation cover is highly inefficient mainly because the cardinality of ${\mathcal{S}}$ is too large. It is possible to reduce the size of ${\mathcal{S}}$ by selecting tessellations that contain at least one maximal clique. This algorithm is not exhaustive but is efficient enough when we use greedy algorithms for the set cover problem.

\subsubsection*{Class with $T(\Gamma) = \chi(K(\Gamma))$.}

\begin{figure}[!h]
\centering
 \includegraphics[width=5cm]{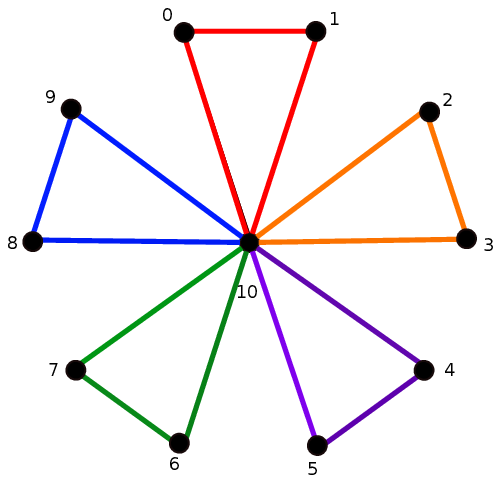}
 \caption{\label{fig:n} The class of windmill graphs considered in Proposition~\ref{prop:windmill} is a tight graph class with respect to the upper bound of Proposition~\ref{prop:upper}. The windmill $\Gamma$ has $5$ triangles, its clique graph $K(\Gamma)$ is the complete graph with $5$ vertices, and the tessellation number $T(\Gamma)$ is $5$.}
 %\caption{\label{fig:n} Windmill graph. Tight class with respect to the upper bound of Proposition~\ref{prop:upper} viewed in Proposition~\ref{prop:windmill}. The clique graph of this graph with $n+1$ vertices is the complete graph with $n$ vertices.}
\end{figure}

\newpage

\subsubsection*{Class with $\frac{T(\Gamma)}{\chi(K(\Gamma))} = \frac{1}{2}$.}

\begin{figure}[!h]
\centering
 \includegraphics[width=5cm]{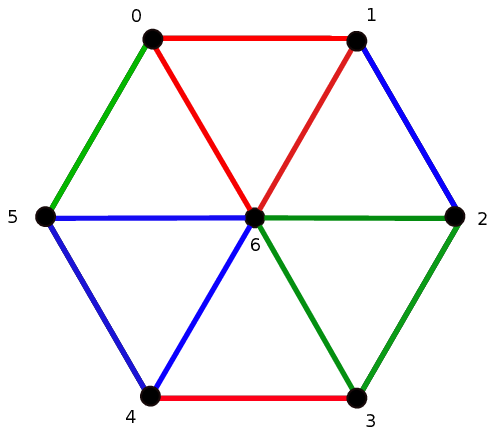}
 \caption{\label{fig:n2} The class of wheel graphs. In the example, the wheel $\Gamma=W_6$ has $7$ vertices, its clique graph $K(\Gamma)$ is the complete graph with $6$ vertices, and the tessellation number $T(\Gamma)$ is $3$. A $3$-tessellation cover is highlighted by using $3$ colors.}
%  \caption{\label{fig:n2} Wheel graph. The clique graph of this graph with $n+1$ vertices is the complete graph with $n$ vertices.}
\end{figure}

\subsubsection*{Class with $\frac{T(\Gamma)}{\chi(K(\Gamma))} = \frac{1}{3}$.}

\begin{figure}[!h]
\centering
 \includegraphics[width=5cm]{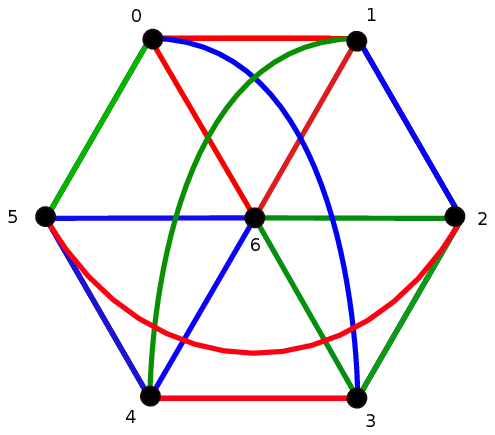}
 \caption{\label{fig:n3} In the example, the graph $\Gamma$ contains the wheel $W_6$ as a spanning subgraph, its clique graph $K(\Gamma)$ is the complete graph with $9$ vertices, and the tessellation number $T(\Gamma)$ is $3$. A $3$-tessellation cover is highlighted by using 3 colors.}
%  \caption{\label{fig:n3} The clique graph of this graph with $n+1$ vertices is the complete graph with $\frac{3n}{2}$ vertices.}
\end{figure}

\newpage

\subsubsection*{Class with $\frac{T(\Gamma)}{\chi(K(\Gamma))} = \frac{1}{4}$.}

\begin{figure}[h]
\centering
 \includegraphics[width=12cm]{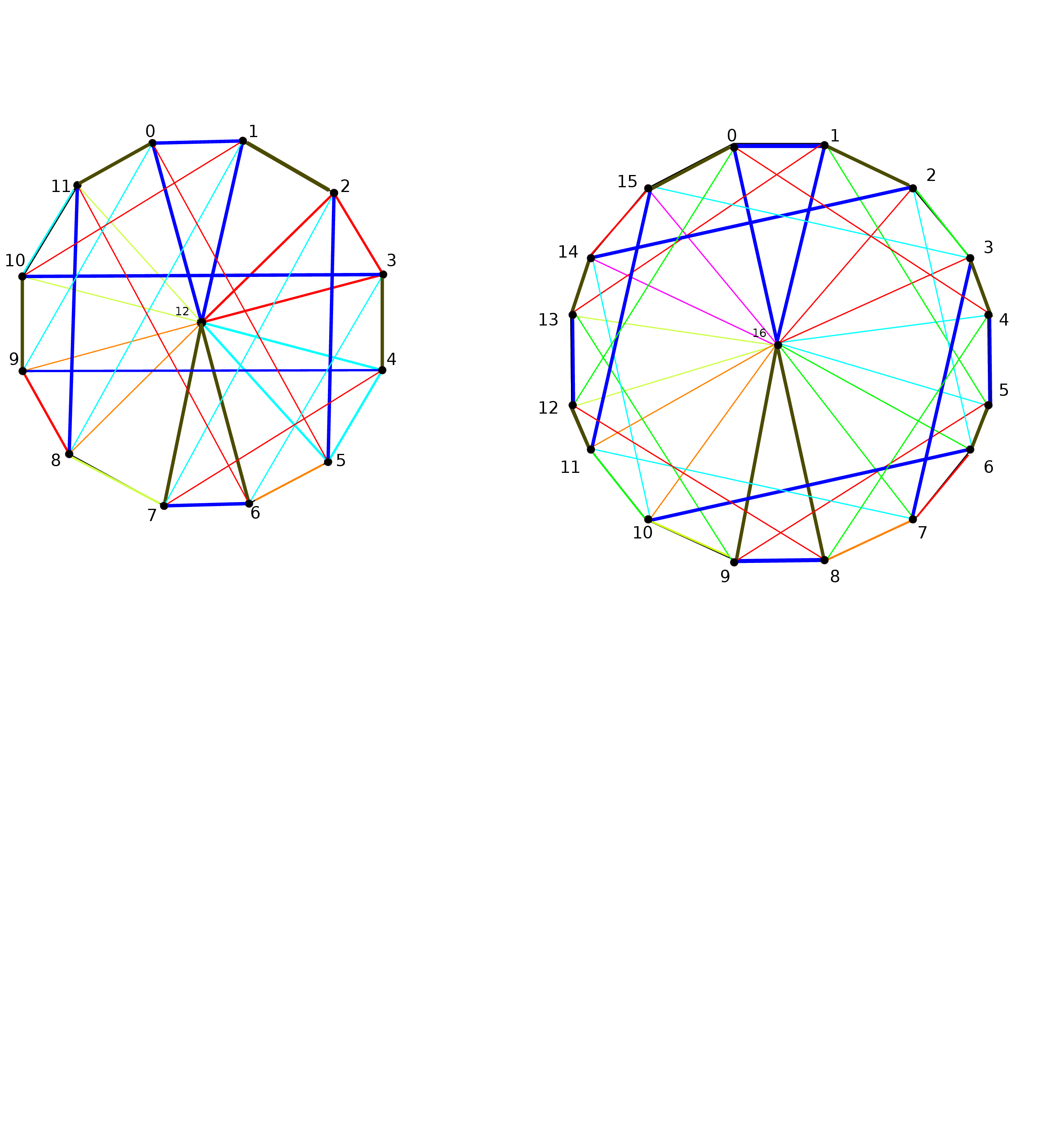}
 \caption{\label{fig:n4} In the examples, the graphs $\Gamma$ contain respectively the wheels $W_{12}$ and $W_{16}$ as a spanning subgraph. The clique graphs $K(\Gamma)$ are respectively the complete graphs $K_{24}$ and $K_{32}$, and the tessellation numbers $T(\Gamma)$ are respectively $6$ and $8$. The corresponding minimum tessellations are highlighted.}%The clique graph of this graph with $n+1$ vertices is the complete graph with $2n$ vertices.}%\label{fig:n4} Each color represents the main edges of a tessellation. The graph starts of a wheel graph $W_{4q}$, we add $4q$ new edges. The clique graph of this new graph is the complete graph with $8q$ vertices, and we have a cover with $2q$ tessellations.}
\end{figure}

% \newpage

\subsubsection*{Class with $T(\Gamma) = \sqrt{\chi(K(\Gamma))}$.}

\begin{figure}[!h]
\centering
 \includegraphics[width=5cm]{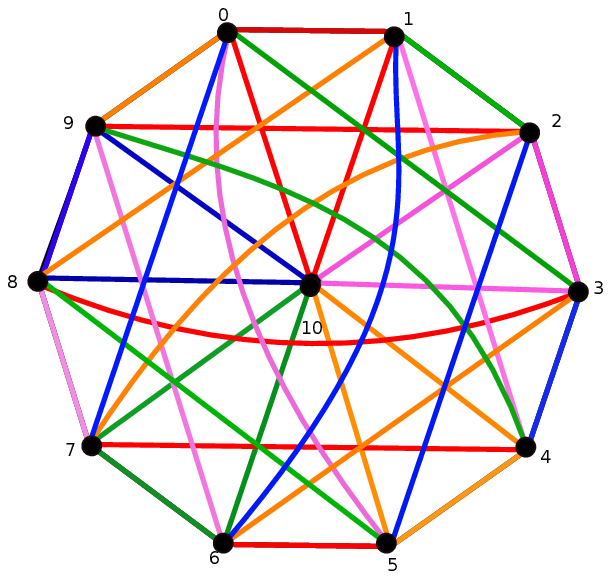}
 \caption{\label{fig:n3} In the example, the graph $\Gamma$ contains the wheel $W_{10}$ as a spanning subgraph, its clique graph $K(\Gamma)$ is the complete graph with $25$ vertices, and the tessellation number $T(\Gamma)$ is $5$. A $5$-tessellation cover is highlighted by using $5$ colors.}
 
 %Extended wheel graph. The clique graph of this graph with $2n+1$ vertices is the complete graph with $n^2$ vertices.}
\end{figure}

\newpage

\subsubsection*{Graph with $T(\Gamma) < \sqrt{\chi(K(\Gamma))}$.}

\begin{figure}[!h]
\centering
 
 \begin{subfigure}{0.4\textwidth}
   \includegraphics[width=4.5cm]{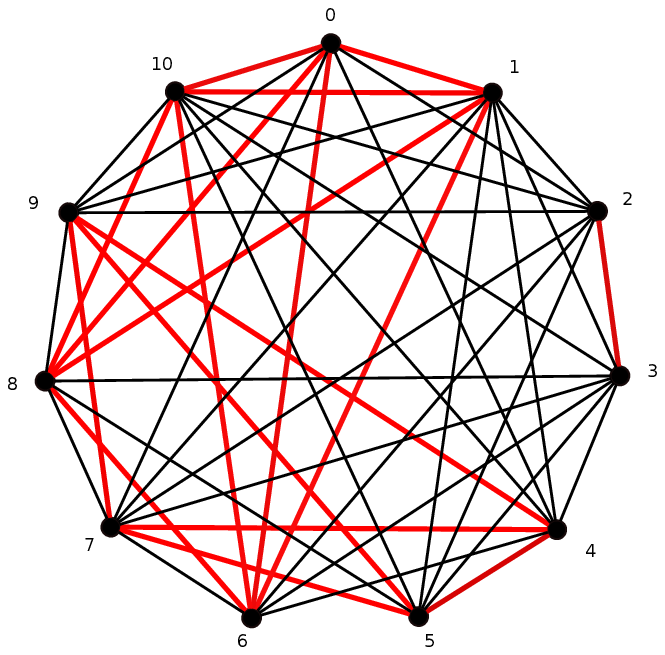}
 \end{subfigure}
\begin{subfigure}{0.4\textwidth}
   \includegraphics[width=4.5cm]{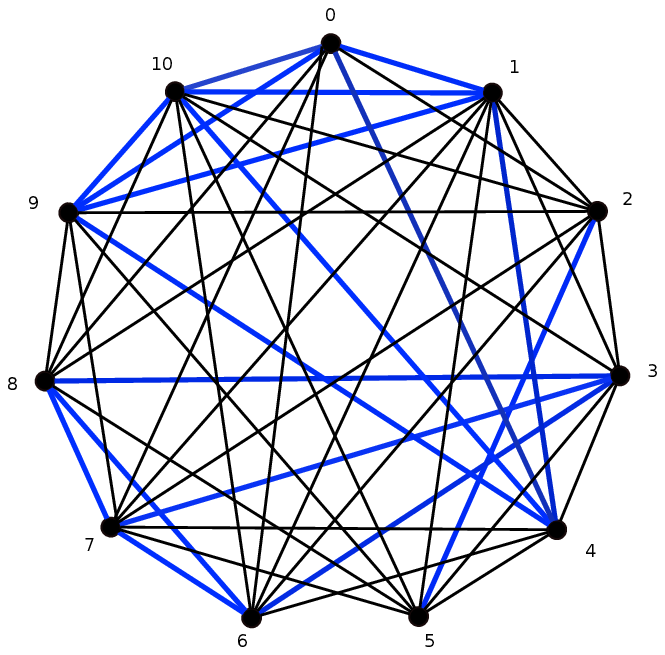}
 \end{subfigure}
 \begin{subfigure}{0.4\textwidth}
   \includegraphics[width=4.5cm]{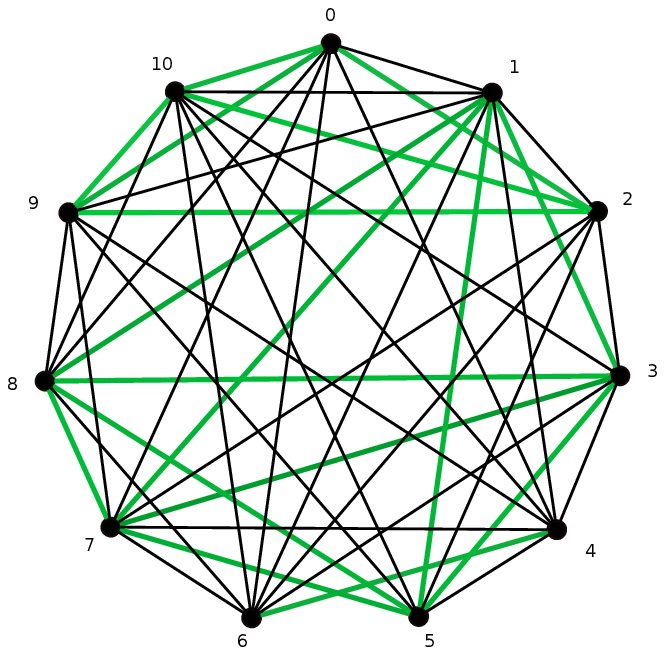}
 \end{subfigure}
 \begin{subfigure}{0.4\textwidth}
   \includegraphics[width=4.5cm]{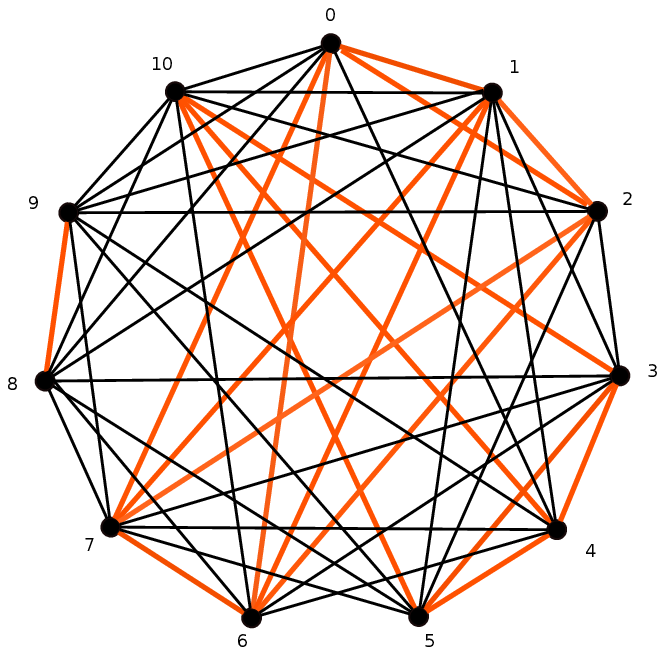}
 \end{subfigure}
 
 \caption{\label{fig:sqrtMenor} Tessellations of a graph $\Gamma$ with $T(\Gamma)=4$ and $\chi(K(\Gamma)) = 30$.}
\end{figure}

\subsubsection*{Class $E_{3,n}$ with $T(\Gamma)=3$ and $\chi(K(\Gamma)) = 3n+3$.}

\begin{figure}[!h]
\centering
 \includegraphics[width=5cm]{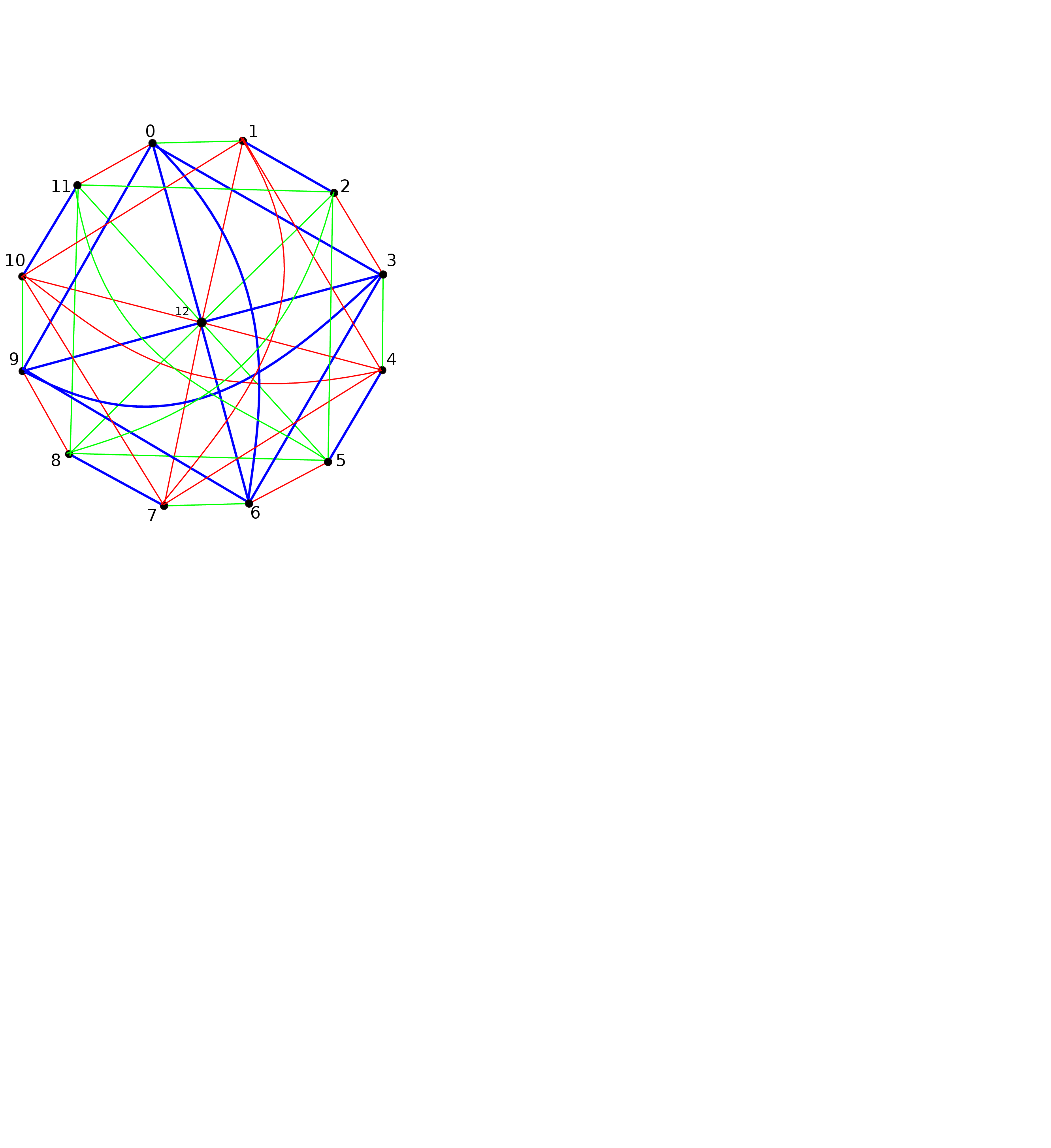}
 \caption{\label{fig:3n} The class of extended wheel graphs considered in Proposition~\ref{lemma2}. The extended wheel graph $\Gamma = E_{3,4}$ has three $5$-cliques and twelve $3$-cliques, hence its clique graph $K(\Gamma)$ is the complete graph with $15$ vertices.}
 
 %Extended wheel graph. The clique graph of this graph with $2n+1$ vertices is the complete graph with $n^2$ vertices.}
\end{figure}

\end{document}